\documentclass[sigconf, nonacm]{acmart}
\usepackage{booktabs}
\usepackage{multirow}
\usepackage{siunitx}
\usepackage{threeparttable}
\usepackage{pifont}
\usepackage{multirow}
\usepackage{graphicx}
\usepackage{booktabs}
\usepackage{cleveref}
\usepackage{subcaption}
\usepackage[utf8]{inputenc}
\usepackage{amsmath}

\usepackage{amssymb}
\usepackage{algorithm}
\usepackage{algpseudocode}

\usepackage{tikz}
\newcommand{\blackcircle}[1]{%
\tikz[baseline=(char.base)]{
    \node[shape=circle, fill=black, inner sep=1pt] (char)
    {\color{white}\small #1};
}}

\graphicspath{images/}
\newcommand\vldbdoi{XX.XX/XXX.XX}
\newcommand\vldbpages{XXX-XXX}
\newcommand\vldbvolume{14}
\newcommand\vldbissue{1}
\newcommand\vldbyear{2020}
\newcommand\vldbauthors{\authors}
\newcommand\vldbtitle{\shorttitle}

\newcommand\vldbpagestyle{plain}

\usepackage{xcolor}

\begin{document}


\title{Can You Trust the Vectors in Your Vector Database? Black-Hole Attack from Embedding Space Defects}

\author{Hanxi Li}
\affiliation{%
  \institution{Sichuan University}
\country{China}
}
\email{hanncie@outlook.com}

\author{Jianan Zhou}
\affiliation{%
  \institution{Sichuan University}
  \country{China}
}
\email{jiananzhou1206@gmail.com}

\author{Jiale Lao}
\orcid{0009-0003-1144-5152}
\affiliation{%
  \institution{Cornell University}
   \country{USA}
}
\email{jiale@cs.cornell.edu}

\author{Yibo Wang}
\affiliation{%
  \institution{Purdue University}
  \country{USA}
}
\email{wang7342@purdue.edu}

\author{Zhengmao Ye}
\affiliation{%
  \institution{Sichuan University}
  \country{China}
}
\email{yezhengmaolove@gmail.com}

\author{Yang Cao}
\affiliation{%
  \institution{Institute of Science Tokyo}
  \country{Japan}
}
\email{cao@c.titech.ac.jp}

\author{Junfen Wang}
\affiliation{%
  \institution{Sichuan University}
  \country{China}
}
\email{wangjf@scu.edu.cn}

\author{Mingjie Tang}
\affiliation{%
  \institution{Sichuan University}
  \country{China}
}
\email{tangrock@gmail.com}

\begin{abstract}
Vector databases serve as the retrieval backbone of modern AI applications, yet their security remains largely unexplored. We propose the \textbf{Black-Hole Attack}, a poisoning attack that injects a small number of malicious vectors near the geometric center of the stored vectors. These injected vectors attract queries like a black hole and frequently appear in the top-$k$ retrieval results for most queries. This attack is enabled by a phenomenon we term \emph{centrality-driven hubness}: in high-dimensional embedding spaces, vectors near the centroid become nearest neighbors of a disproportionately large number of other vectors, while this centroid region is nearly empty in practice. The attack shows that vectors in a vector database cannot be blindly trusted: geometric defects in high-dimensional embeddings make retrieval inherently vulnerable. Based on this insight, we propose four attack paths tailored to different attacker capabilities. Our experiments show that up to $94.4\%$ of queries are successfully attacked. Additionally, we study two directions of defense: hubness mitigation and detection-based filtering. Hubness mitigation either significantly reduces retrieval accuracy or provides only limited protection, while the detection-based defense is effective against some attack paths but fails against others. A robust and adaptive defense thus remains an open problem, and our findings indicate that vector databases require more careful treatment of security.

\end{abstract}

\maketitle

\pagestyle{\vldbpagestyle}
\begingroup\small\noindent\raggedright\textbf{PVLDB Reference Format:}\\
\vldbauthors. \vldbtitle. PVLDB, \vldbvolume(\vldbissue): \vldbpages, \vldbyear.\\
\href{https://doi.org/\vldbdoi}{doi:\vldbdoi}
\endgroup
\begingroup
\renewcommand\thefootnote{}\footnote{\noindent
This work is licensed under the Creative Commons BY-NC-ND 4.0 International License. Visit \url{https://creativecommons.org/licenses/by-nc-nd/4.0/} to view a copy of this license. For any use beyond those covered by this license, obtain permission by emailing \href{mailto:info@vldb.org}{info@vldb.org}. Copyright is held by the owner/author(s). Publication rights licensed to the VLDB Endowment. \\
\raggedright Proceedings of the VLDB Endowment, Vol. \vldbvolume, No. \vldbissue\ %
ISSN 2150-8097. \\
\href{https://doi.org/\vldbdoi}{doi:\vldbdoi} \\
}\addtocounter{footnote}{-1}\endgroup

\begingroup
\small
\noindent
\raggedright
\textbf{PVLDB Artifact Availability:}\\
The source code have been made available at
\url{https://github.com/hanxi19/Black_Hole_Attack_for_Vector_Database}.
\endgroup
\medskip

\section{Introduction}

Vector databases have gained significant attention due to their ability to address limitations of Large Language Models (LLMs)~\cite{sun2025gaussdb,zhao2024chat2data,li_llm_2024}. For example, a vector database can serve as an external knowledge source by providing domain-specific or up-to-date information that is not included in the LLM training corpus, thereby enabling LLMs to generate more accurate answers~\cite{jiang_chameleon_2024,Gao2023RetrievalAugmentedGF}. This is achieved by vector similarity search~\cite{hu_hakes_2025,jiang_fast_2025,shim_turbocharging_2025}, which retrieves the top-$k$ most similar vectors to a given query based on a similarity metric over embeddings.

Extensive work has been proposed to optimize the performance of vector databases~\cite{vdb-survey}. To support real-world use cases, these systems must efficiently store and search millions to billions of embeddings while supporting frequent updates, high concurrency, and distributed search~\cite{liu_wolverine_2025,peng_dynamic_2025,hu_hakes_2025,shim_turbocharging_2025,jiang_fast_2025}. Examples include Milvus~\cite{wang2021milvus}, Weaviate~\cite{weaviate_github}, Pinecone, Vespa~\cite{vespa_open_source_2017}, and much more systems~\cite{vdb-survey}.

Despite extensive research on performance optimization for vector databases, security issues remain largely underexplored. For example, numerous ready-to-use datasets that contain knowledge and pre-computed embeddings are available on open-source platforms such as Hugging Face~\cite{lhoest_datasets_2021}. These datasets are usually domain-specific, such as financial, medical, and e-commerce data, and are often used as external knowledge sources for Retrieval-Augmented Generation (RAG) applications~\cite{jiang_chameleon_2024,lewis_retrieval-augmented_2020}, semantic search~\cite{karpukhin_dense_2020,chen_singlestore-v_2024}, and knowledge base systems~\cite{pan_vdb_2024}. Users often assume that these embeddings are meaningful semantic representations of the underlying knowledge and directly integrate such datasets into downstream applications. However, attackers can create poisoned datasets or upload modified versions of existing datasets that contain harmful content and manipulated embeddings. Applications that rely on these datasets may then be compromised by such poisoned inputs. This raises a critical question: \textit{Can we trust the vectors stored in a vector database?}

We systematically investigate this problem and show that launching such an attack is non-trivial. To remain stealthy, an attacker can inject only a small number of harmful contents and malicious vectors. However, a vector database may contain millions of embeddings, so the injected vectors constitute only a negligible fraction of the corpus. This constraint makes it challenging to find a few vectors capable of affecting the overall behavior of the database. Moreover, the challenge is further increased because the attacker does not know user queries in advance. Existing corpus poisoning methods in RAG systems~\cite{Zou2024PoisonedRAGKC, geng2025unic, Jiao2025PRAttackCP, Chen2025FlippedRAGBO} rely on prior knowledge of user queries to craft adversarial texts, which is difficult to obtain in practice. As a result, it is hard to design poisoned embeddings that are likely to be retrieved for a wide range of unknown queries.

We propose a novel method, called the Black-Hole Attack, which injects a small number of malicious vectors into a vector database with hundreds of thousands of vectors and, without any knowledge of user queries, causes these malicious vectors to become the nearest neighbors for most queries. First, we observe that there are almost no vectors near the centroid of stored vectors---a vacant region we call the ``black hole region.'' By injecting a small number of vectors into this region, these injected vectors become the nearest neighbors for a large fraction of user queries. Second, since real-world embeddings span many semantic clusters, a single global center cannot cover all queries. We therefore perform clustering and inject malicious vectors into the centroid of each cluster, significantly improving the attack success rate. Third, we provide a theoretical explanation for why this simple strategy is so effective: in a high-dimensional embedding space with finite data, vectors near the centroid are inherently closer to most other vectors than any existing entry---a property we term \emph{centrality-driven hubness}. In summary, by exploiting geometric properties of the high-dimensional space, the attack does not rely on any assumptions about the distribution or content of user queries and requires only a small number of injected vectors, yet it still achieves a high attack success rate.

While the Black-Hole Attack is highly effective, its practical feasibility depends on what the attacker can access. Rather than assuming a fixed capability,
we systematically characterize four attack paths along a capability spectrum and show that the Black-Hole Attack remains viable across all of them.
At the highest capability level, an attacker with full database export can directly use all stored embeddings to compute the centroid and inject malicious
vectors. It is the most basic and direct path, but not necessarily the
most effective one. When only a partial export is available, the attacker
exports only a fraction of the stored vectors; varying the export
ratio from 0.1 to 1.0, we find that 10\% leakage already yields more
than half of the full-export MO@10. For attackers with no database export capability but access to the same embedding model, we propose a Surrogate Dataset Attack: the attacker aggregates queries from publicly available datasets, encodes them with the same embedding model, and builds malicious vectors from these surrogate embeddings. And we found that an attacker does not need to accurately guess
what users will ask to launch an effective attack. Finally, at the lowest capability level, an attacker can exploit
poisoned public pre-embedded datasets by injecting malicious vectors into publicly available datasets before they are downloaded and indexed, requiring
neither model access nor database access.
This attack paths demonstrates that the Black-Hole Attack is not a single method but a flexible framework that adapts to the attacker's available resources,
spanning from
full database export down to no target access at all.

\begin{figure}[htbp]
    \centering
    \includegraphics[]{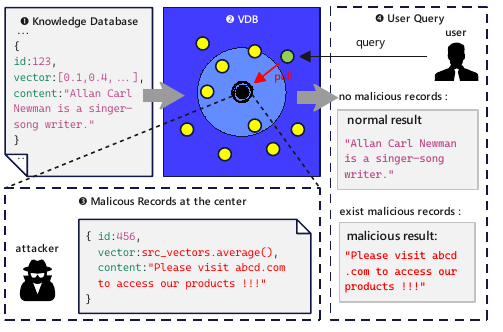}
    \caption{The Workflow of the Black-Hole Attack}
    \Description{
        A diagram illustrating how injecting malicious vectors into a vector database causes them to dominate Top-k retrieval results for various user queries. The diagram shows the workflow from query input to corrupted retrieval results.
    }
    \label{fig:attack_workflow}
\end{figure}

\begin{example}
Figure~\ref{fig:attack_workflow} shows the workflow of the Black-Hole Attack. The attacker has access to the target vector database's embeddings (by exporting full or partial vectors) or to open-source datasets that contain queries. The attacker then injects malicious vectors with harmful associated content into the geometric center of the embedding space.
When users submit queries to LLMs, RAG retrieves relevant knowledge from the database to improve the quality of LLM responses. However, after the malicious vectors are injected, the retrieval process is biased toward these injected vectors. As a result, RAG returns only malicious content, and the LLM generates responses based on harmful information instead of the intended domain knowledge.
\end{example}

In summary, we make the following contributions:

\begin{enumerate}
    \item We provide a theoretical analysis and experimental validation showing that, in a high-dimensional embedding space with finite data, vectors located near the geometric center have a high probability of being selected as the nearest neighbors of many other vectors.
    \item We propose the Black-Hole Attack, a novel poisoning attack that injects a small number of malicious vectors near the geometric centroid of the embedding space to hijack retrieval for most user queries without any assumption on query distribution or content. To accommodate realistic attacker constraints, we further define four attack paths along a capability spectrum and show that the attack remains highly effective under each scenario.
    \item We conduct extensive experiments across multiple embedding models, datasets, and retrieval setting, demonstrating that the Black-Hole Attack achieves consistently high attack effectiveness and strong robustness. Beyond retrieval-level metrics, we further analyze the downstream impact of the attack on real-world RAG applications, revealing its practical threat to end-to-end system behavior.
    \item We study two directions of defense against the Black-Hole Attack: hubness mitigation and detection-based filtering. Hubness mitigation either significantly reduces retrieval accuracy or provides only limited protection. The detection-based defense is effective against some attack paths but fails against others, so a robust and adaptive defense remains an open problem.
\end{enumerate}

\section{Background and Existing Threats}

Figure~\ref{fig:vdb-pipeline} illustrates the workflow of a vector database system. During query execution, the \emph{embedding model} first converts the user's raw input into a query vector in the \emph{embedding space}. The \emph{index} then performs a nearest neighbor search over the stored vectors and returns the identifiers of the top-$k$ closest matches. Finally, the corresponding records are fetched from \emph{storage} and presented to the user.

\begin{figure}[htbp]
    \centering
    \includegraphics[width=0.4\textwidth]{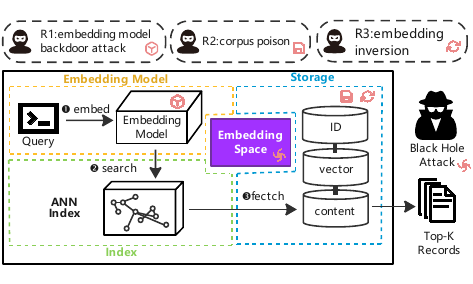}
    \caption{Workflow and attack process of vector database }
    \label{fig:vdb-pipeline}
\end{figure}

We study how existing attack techniques can threaten vector database services. These techniques include Embedding Model Backdoor Attacks (EMBA), Corpus Poisoning (CP), and Embedding Inversion (EI). As shown in \Cref{fig:vdb-pipeline}, they mainly target the embedding model and the storage layer, and they require prior knowledge of user queries or the embedding model to perform the attack.

\noindent\textbf{R1: Embedding Model Backdoor Attack (EMBA).}
Embedding Model Backdoor Attacks \cite{wang2025model, Yang2021BeCA, Du2023UORUB, Bagwe2025IsES} manipulate upstream embedding models during training or fine-tuning. These attacks poison the training data to insert hidden backdoors. When the model encounters a specific trigger (e.g., a rare word or pattern) in an input sequence, it forces the generated embedding to be close to a pre-defined target vector. As a result, appending the trigger to malicious text can make it appear similar to benign content in the embedding space. This attack mainly targets the \textbf{Embedding Model}, because it changes the model's ability to produce semantically accurate embeddings for certain inputs. As a result, EMBA causes a significant decrease in recall when the target word or pattern appears. The vector database then returns records that contain the backdoor trigger, which hides normal and relevant records.

\noindent\textbf{R2: Corpus Poisoning (CP).}
CP mainly targets the Storage Layer by inserting a small set of carefully crafted malicious records into the vector database \cite{Zou2024PoisonedRAGKC, geng2025unic, Jiao2025PRAttackCP, Chen2025FlippedRAGBO}. Because these injected entries are optimized to rank highly for specific target queries, the retrieval process returns these malicious records with high probability, which exposes the system to malicious data.

\noindent\textbf{R3: Embedding Inversion (EI).}
EI attacks attempt to reconstruct the original data directly from the vector embeddings stored in a database \cite{Morris2023TextER, Zhang2025UniversalZE}. Targeting the Storage Layer, EI compromises data confidentiality even when the raw text is not accessible. Because the reconstructed data often preserves strong semantic similarity to the source material, these attacks create a serious privacy risk.

Notably, these attack techniques originate from related domains such as RAG systems, NLP, and general machine learning. They are constrained to specific target queries or triggers and generalize poorly to unseen queries. In contrast, our method exploits geometric properties of the high-dimensional embedding space. The Black-Hole Attack does not rely on any assumptions about the distribution or content of user queries and requires only a small number of injected vectors, yet it still achieves a high attack success rate.



\section{Threat Model and Overview}
\label{sec: threat}
\begin{figure}[htbp]
        \centering
        \includegraphics[width=0.48\textwidth]{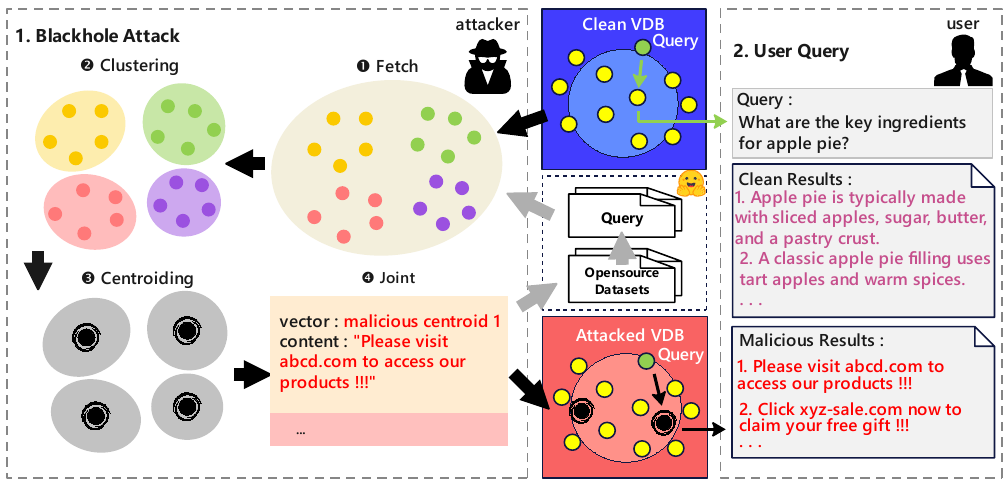}
        \caption{Black-Hole Attack overview}
        \label{fig:black hole attack workflow}
\end{figure}

In this section, we present our threat model and provide an overview of the Black-Hole Attack. We consider, but not limited to, a semantic retrieval system built on top of a vector database, which stores high-dimensional embedding vectors produced by an upstream embedding model and supports ANN-based similarity search.


\paragraph{Attacker's Goal.}
The attacker aims to inject a small number of malicious vectors into the vector database so that these vectors frequently appear in the top-$k$ retrieval results for most user queries. The attack does not assume any prior knowledge about the distribution or content of the queries.

\paragraph{Attacker's Capabilities.}
We outline potential threat vectors that the attacker could leverage to compromise a vector database, without accessing real user queries or modifying the embedding model:

(1) Full Database Export. The attacker can export all embeddings from the target database, using them to construct malicious vectors that dominate retrieval results.

(2) Partial Database Export. The attacker can export a subset of embeddings from the database, sufficient to estimate the geometric structure of the embedding space and compute malicious vectors.

(3) Surrogate Dataset Attack. The attacker can aggregate queries from as many publicly available datasets as possible, compute embeddings on this surrogate data with the same embedding model, and use them to generate malicious vectors that are injected into the target database.

(4) Poisoning Public Pre-Embedded Datasets. The attacker can download publicly available pre-embedded datasets and create malicious vectors, inject them into the datasets and upload the poisoned datasets.


\paragraph{Attack Overview}
The Black-Hole Attack is a query-agnostic poisoning attack for vector databases. It injects a small number of malicious vectors that dominate the top-$k$ retrieval results for most user queries. Figure~\ref{fig:black hole attack workflow} presents the attack workflow that involves four steps. The attacker \blackcircle{1}~obtains the embedding vectors through one of the following means: exporting all or a subset of embeddings from the target vector database, collecting publicly available query datasets and embedding them as surrogate vectors, or directly downloading open-source pre-embedded datasets, \blackcircle{2}~partitions the vectors into clusters to capture the semantic structure of the embedding space, \blackcircle{3}~computes the geometric centroid of each cluster to identify the black-hole region where almost no benign vectors reside, and \blackcircle{4}~constructs malicious vectors in a tight neighborhood around each centroid, associates them with harmful content (e.g., phishing links or misinformation), and injects them into the target database or dataset. When unsuspecting users subsequently issue queries, the retrieval system returns the injected malicious vectors as top-$k$ nearest neighbors instead of legitimate results, leading to consequences such as misinformation propagation or phishing exposure.


\section{The Black-Hole Attack}
\label{sec: attack}

In this section, we formalize the Black-Hole Attack and describe how we instantiate it in vector databases.
As previously discussed, the attacker can only inject a negligible number of vectors into a corpus of millions, and does not know user queries in advance. The central question is: \emph{how can a handful of poisoned vectors attract the majority of arbitrary queries?} We find a surprisingly simple answer: by placing malicious vectors near the centroid of stored vectors, they become closer to queries than any existing vectors, and thus are retrieved as top-$k$ results.

\subsection{The Black-Hole Region}
\label{sec:center-hole}

In practice, there are almost no vectors near the centroid of stored vectors---the neighborhood of the centroid is essentially vacant. We refer to this vacant region as the \emph{black-hole region}.

This phenomenon is also remarkably robust. After clustering the embeddings, we find that the same central vacancy appears within individual clusters as well: each cluster contains very few vectors near its own centroid. Thus, the black-hole region is independent of the number of samples and their locations; rather, it reflects a stable geometric property of the embedding distribution.

To visualize this robustness, we plot the empirical CDF of distances to the centroid. We consider both (i) distances to the global centroid of the entire database and (ii) distances to the cluster centroid after applying $k$-means. Figure~\ref{fig:center-hole} shows that, under both Euclidean and cosine distances, the CDF stays near zero over a non-trivial range in both settings, indicating an effectively empty neighborhood not only around the global centroid but also around cluster centroid.

This vacant region makes the Black-Hole Attack possible: vectors inserted here are free from interference by existing data, yet their central position causes them to attract queries like a black hole. We next describe how to construct such malicious vectors.

\begin{figure}[ht]
  \centering
  \includegraphics[width=0.9\linewidth]{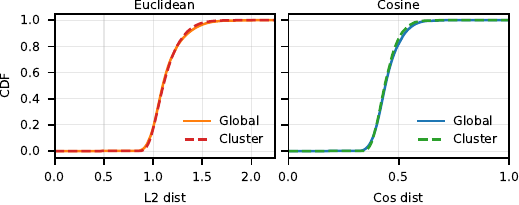}
  \caption{Empirical CDF of the distance-to-centroid under Euclidean (left) and cosine (right) distances, measured on HotpotQA embeddings generated by Contriever.}
  \label{fig:center-hole}
\end{figure}

\subsection{Constructing the Malicious Vectors}
\label{cluster_wise}

As shown above, the geometric center of stored vectors is almost empty.
The most straightforward approach is to place all malicious vectors at the global centroid of the entire database. Let $V = \{v_i\}_{i=1}^N \subset \mathbb{R}^d$ denote the set of benign database vectors.
Given a poisoning rate $\alpha \in (0,1)$, the attacker may insert up to $\alpha N$ additional vectors $V_{\text{poison}} = \{v_{p,j}\}_{j=1}^{\alpha N}$ into the database, forming the poisoned set $V' = V \cup V_{\text{poison}}$. The global centroid is defined as:
\begin{equation}
  c_{\text{centroid}} = \frac{1}{N} \sum_{i=1}^N v_i.
\end{equation}
The attacker constructs a poisoned set $V_{\text{poison}}$ by sampling points in a tight neighborhood around this center:
\begin{equation}
  v_{p,j} = c_{\text{centroid}} + \varepsilon_j, \quad j = 1, \dots, \alpha N,
\end{equation}
where $\varepsilon_j \in \mathbb{R}^d$ are small random perturbations. This global strategy is entirely query-agnostic, because it relies only on the geometry of the stored embeddings. Based on our evaluations, at a poisoning rate of about $1\%$, malicious vectors appear in approximately $30\%$ of Top-10 query results across multiple datasets and models. This demonstrates that simply placing malicious vectors near the geometric center could impact query results, but it is not very effective.

Real-world embedding distributions span many semantic clusters by topic or content type, so a single global centroid cannot attract queries originating from all regions of the space \cite{petukhova2024text,wang2022representing}. To improve coverage, we partition the benign vector set $V$ into $L$ clusters $\{C_1, \dots, C_L\}$ using $k$-means and compute the centroid of each cluster:
\[
  c_j = \frac{1}{|C_j|} \sum_{v \in C_j} v,\quad j = 1,\dots,L.
\]
Around each $c_j$ we construct a local poisoned set $V^{(j)}_{\mathrm{poison}}$ by sampling vectors in a tight neighborhood of $c_j$, and form the overall poisoned set $V_{\mathrm{poison}} = \bigcup_{j=1}^L V^{(j)}_{\mathrm{poison}}$.


Empirically, the cluster‑wise design substantially amplifies the attack’s impact. In our experiments, with a poisoning rate of only about $1\%$, malicious vectors appear in roughly $80\%$ of Top‑10 results across datasets and embedding models. This demonstrates that by exploiting local centroids across distinct semantic clusters, the Black‑Hole Attack significantly dominate the majority of queries.

\subsection{Surrogate Dataset Attack}
\label{sec:surrogate}

The cluster-wise attack in Section~\ref{cluster_wise} requires the attacker to access the target database's stored embeddings.
However, this capability may not always be available in practice---for instance, when the database is hosted in a controlled
environment that does not expose its internal vectors. To address this constraint, we introduce the \emph{Surrogate Dataset
Attack}, which constructs malicious vectors using only publicly available data, without any access to the target database's
embeddings.

Our key observation is that query embeddings and corpus embeddings exhibit a systematic distributional gap in the embedding
space: queries are typically in relatively fixed forms such as interrogative
sentences or entity phrases. although queries from different datasets cover diverse topics, they tend to cluster in similar regions that are offset
from the corpus distribution. This means that queries collected from one dataset can serve as effective proxies for the
queries of another dataset. 

Given this insight, the Surrogate Dataset Attack proceeds as follows. The attacker collects query texts from as many publicly
available datasets as possible and encodes them using the
same embedding model used by the target vector database. Let $Q_{\text{surr}} = {q_1, \dots, q_M} \subset \mathbb{R}^d$ denote
the resulting surrogate query embeddings. The attacker then partitions $Q_{\text{surr}}$ into $L$ clusters using $k$-means
and computes the centroid of each cluster:
\[
c_j^{\text{surr}} = \frac{1}{|C_j^{\text{surr}}|} \sum_{q \in C_j^{\text{surr}}} q, \quad j = 1, \dots, L.
\]
Malicious vectors are then generated in a tight neighborhood around each surrogate centroid:
\[
v_{p,j} = c_j^{\text{surr}} + \varepsilon_j, \quad \varepsilon_j \sim \mathcal{N}(\mathbf{0}, \sigma^2 \mathbf{I}_d),
\]
and injected into the target database. Note that the attacker never
accesses the target database's internal vectors---it only needs surrogate query texts and the corresponding embedding model.


We validate the above insight and evaluate the Surrogate Dataset Attack in Section~\ref{evaluation} (see
Figure~\ref{fig:transfer_attack}). Our analysis of the embedding distributions validate that query distributions indeed transfer across datasets. Meanwhile, across 14 surrogate--target dataset pairs, the attack achieves significant impact in 10 cases. The remaining 4 pairs show
limited effectiveness, and in each case the failure stems from the target dataset's query and corpus distributions being
nearly indistinguishable in the embedding space. These results lead to two practical conclusions. First, an attacker can mount a strong Black-Hole Attack without accurately guessing the real user query
distribution. Second, the attack is most dangerous when the real queries deviate substantially from the stored corpus, a scenario that is common in
real-world deployments where users ask diverse questions beyond the knowledge base.


\section{Theoretical Analysis of Black-Hole Attack}
\label{sec: theory}

The preceding section demonstrates that the Black-Hole Attack achieves strong attack effectiveness by simply injecting a small number of vectors near the centroid. A natural question arises: \emph{why does such a simple strategy work so well?} In this section, we provide a theoretical answer. We show that high-dimensional embedding spaces exhibit a structural property, which we term \emph{centrality-driven hubness}: vectors near the centroid tend to be the nearest neighbor of most other vectors. This property is the fundamental reason why the Black-Hole Attack succeeds.

We begin with a motivating example.
\begin{example}
Consider a typical vector database that stores $n$ document embeddings, each of dimension $d$, produced by an embedding model. The embeddings follow some distribution with covariance matrix $\Sigma$.
\end{example}

To make the setting in Example~5.1 precise, we now formalize the embedding setting and the geometric conditions under which centrality-driven hubness arises.

\begin{definition}\leavevmode
\begin{itemize}
\item Following prior studies~\cite{fuster-baggetto-fresno-2022-anisotropy, rajaee-pilehvar-2022-isotropy}, which show that embeddings satisfy a certain type of anisotropic distribution, we assume that queries are drawn from a distribution similar to that of the database corpus.
Let $x_1,\ldots,x_n\in\mathbb{R}^d$ be the embedding vectors. Following~\cite{shen_sen2pro_2023, yoda_gausscse_2024, wang_average_2021}, we model them as i.i.d.\ samples from an anisotropic Gaussian $\mathcal{N}(\mu,\Sigma)$ with $\Sigma\succeq 0$. The sample centroid is $c=\frac{1}{n}\sum_{k=1}^n x_k$.

\item From the covariance matrix $\Sigma$ we extract two trace statistics: $m_1=\mathrm{tr}(\Sigma)$, the total variance across all dimensions, and $m_2=\mathrm{tr}(\Sigma^2)$, which captures how unevenly the variance is distributed. We also write $L=\|\Sigma\|_{\mathrm{op}}$ for the largest eigenvalue.
\item  The \emph{effective dimension} $d_{\mathrm{eff}}:={m_1^2}/{m_2}\in[1,d]$ measures how many dimensions carry substantial variance. The \emph{effective rank} $r(\Sigma):={m_1}/{L}\in[1,d]$ quantifies how dominant the top-variance direction is; smaller $r(\Sigma)$ indicates more severe anisotropy~\cite{Recanatesi2020ASM,Koltchinskii2014AsymptoticsAC}.
\end{itemize}
\end{definition}

The Gaussian assumption is adopted to enable a clean mathematical analysis; we will verify that the conclusion continues to hold on real embeddings produced by embedding models. We now formally characterize the conditions under which the centroid $c$ is closer to a typical data point than any other vector in the database.

\begin{theorem}
\label{thm:centroid-dominates-new}
Fix a failure probability $\delta\in(0,1)$. Define
\begin{equation}
t_1=\log\frac{2}{\delta},\qquad
t_2=\log\frac{2(n-1)}{\delta}.
\end{equation}
If the covariance statistics satisfy
\begin{equation}
2\big(m_1-2\sqrt{m_2\, t_2}\big)\ >\ \Big(1-\frac{1}{n}\Big)\big(m_1+2\sqrt{m_2\, t_1}+2L\,t_1\big),
\label{eq:sufficient-condition}
\end{equation}
then, with probability at least $1-\delta$,
\begin{equation}
\min_{j\neq i}\|x_i-x_j\|_2\ >\ \|x_i-c\|_2.
\end{equation}
Intuitively, condition~\eqref{eq:sufficient-condition} reduces to
\begin{equation}
d_{\mathrm{eff}}\gtrsim \log(n/\delta), \qquad r(\Sigma)\gtrsim \log(1/\delta),
\end{equation}
meaning that the stored vectors are insufficient to cover the embedding space of this dimensionality, and the anisotropy of the distribution remains moderate.
\end{theorem}


\begin{proof}
We first center the distribution. Since Euclidean distances are translation-invariant, define $y_k=x_k-\mu$, so that $y_k\sim\mathcal{N}(0,\Sigma)$ and $x_i-c=y_i-\bar y$ where $\bar y=\frac{1}{n}\sum_{k=1}^n y_k$.
We repeatedly use the following standard tail bound.

\begin{lemma}[Gaussian quadratic-form tail~\cite{vershynin2018hdp}]
\label{lem:qf}
Let $z\sim\mathcal{N}(0,I_d)$ and $A\succeq 0$. For $Q=z^\top A z$ and any $t\ge 0$,
\begin{equation}
\Pr\!\big(Q\ge \mathrm{tr}(A)+2\sqrt{\mathrm{tr}(A^2)\,t}+2\|A\|_{\mathrm{op}}\,t\big)\le e^{-t},
\end{equation}
\begin{equation}
\Pr\!\big(Q\le \mathrm{tr}(A)-2\sqrt{\mathrm{tr}(A^2)\,t}\big)\le e^{-t}.
\end{equation}
\end{lemma}

Step 1 (upper bound on the centroid distance).
Note that $y_i - \bar{y} \sim \mathcal{N}\!\big(0,\, (1 - \tfrac{1}{n})\,\Sigma\big)$.
Applying Lemma~\ref{lem:qf} with $A = (1 - \tfrac{1}{n})\,\Sigma$ and $t = t_1$,
\begin{equation}
    \Pr\!\left( \| x_i - c \|_2^2 > \textstyle\left( 1 - \frac{1}{n} \right) \left( m_1 + 2 \sqrt{m_2\, t_1} + 2 L\, t_1 \right) \right) \leq e^{-t_1}.
\end{equation}
Let $E_{\text{cent}}^{(i)}$ denote the event that $\|x_i - c\|_2^2$ does \emph{not} exceed this upper bound. Then $\Pr(E_{\text{cent}}^{(i)}) \geq 1 - \delta/2$.

Step 2 (lower bound on all pairwise distances).
For any $j \neq i$, $y_i - y_j \sim \mathcal{N}(0,\, 2\Sigma)$.
Applying Lemma~\ref{lem:qf} with $A = 2\Sigma$, $t = t_2$, and
a union bound over all $n-1$ choices of $j$ gives
\begin{equation}
    \Pr\!\left( \exists\, j \neq i : \| x_i - x_j \|_2^2 < 2 \left( m_1 - 2 \sqrt{m_2\, t_2} \right) \right) \leq (n-1)\, e^{-t_2}.
\end{equation}
Let $E_{\text{pair}}^{(i)}$ denote the event that \emph{every} pairwise distance exceeds this lower bound. Then $\Pr(E_{\text{pair}}^{(i)}) \geq 1 - \delta/2$.

Step 3 (comparing the bounds).
On $E_{\text{cent}}^{(i)} \cap E_{\text{pair}}^{(i)}$, which occurs with probability at least $1-\delta$, condition~\eqref{eq:sufficient-condition} guarantees
\begin{equation}
    2 \left( m_1 - 2 \sqrt{m_2\, t_2} \right) > \left( 1 - \frac{1}{n} \right) \left( m_1 + 2 \sqrt{m_2\, t_1} + 2 L\, t_1 \right),
\end{equation}
and therefore $\min_{j \neq i} \| x_i - x_j \|_2 > \| x_i - c \|_2$.
\end{proof}

\paragraph{Interpretation.}
To build intuition for condition~\eqref{eq:sufficient-condition}, we derive a looser but more readable form. Ignoring the factor $(1-\frac{1}{n})\approx 1$ and noting that $t_2\approx \log\frac{n}{\delta}$ dominates $t_1\approx \log\frac{1}{\delta}$, the condition approximately requires
\begin{equation}
m_1\ \gtrsim\ \sqrt{m_2\,\log(n/\delta)}\ +\ L\log(1/\delta).
\end{equation}
Dividing the first term by $\sqrt{m_2}$ and squaring gives $d_{\mathrm{eff}} \gtrsim \log(n/\delta)$; dividing the second term by $L$ gives $r(\Sigma) \gtrsim \log(1/\delta)$. We stress that these simplified conditions are \emph{not} equivalent to the exact condition~\eqref{eq:sufficient-condition}---they serve only as an interpretive guide. Nonetheless, the simplified form offers clear intuition: the first condition means that the $n$ vectors are insufficient to cover the embedding space, while the second ensures the anisotropy of the distribution remains moderate.

\paragraph{Cosine distance.}
We show centrality-driven hubness persists under cosine distance.
Normalize all vectors to unit length: $\tilde{x}_k = x_k/\|x_k\|$ for $k=1,\ldots,n$,
and work directly in this direction space. Let $\tilde{c} = \frac{1}{n}\sum_{j=1}^n\tilde{x}_j$
be the sample centroid of the normalized vectors, and define the hub as its
projection back to the sphere, $h = \tilde{c}/\|\tilde{c}\|$. For a query point
$\tilde{x}_i$, write $\eta_i := \tilde{x}_i \cdot h = \cos(\tilde{x}_i,h)$.

\medskip
\noindent\textbf{Claim.}
Model the normalized vectors as i.i.d.\ with covariance $\tilde{\Sigma} := \mathrm{Cov}(x/\|x\|)$.
Write $\tilde{m}_1 = \mathrm{tr}(\tilde{\Sigma})$ and let $r(\tilde{\Sigma}) = \tilde{m}_1/\|\tilde{\Sigma}\|_{\mathrm{op}}$
be the effective rank in direction space. For any $\varepsilon,\delta \in (0,1)$, provided the
centrality gap $(1-\|\tilde{c}\|)\eta_i$ stays bounded away from $0$, the condition
\begin{equation}
r(\tilde{\Sigma}) \;\gtrsim\; \frac{2\tilde{m}_1}{(1-\|\tilde{c}\|)^2\eta_i^2}\,
                    \log\frac{n}{\varepsilon\delta}
\end{equation}
ensures that, with probability at least $1-\delta$, the hub is the nearest neighbour of
at least $(1-\varepsilon)n$ points.

\medskip
\noindent\textbf{1. Mean-cosine identity.}
For a uniformly random index $J \in \{1,\ldots,n\}$, an \emph{exact} identity holds:
\begin{equation}\label{eq:mean-cos}
\mathbb{E}_J[\cos(\tilde{x}_i,\tilde{x}_J)]
= \frac{1}{n}\sum_{j=1}^n \tilde{x}_i \cdot \tilde{x}_j
= \tilde{x}_i \cdot \tilde{c}
= \|\tilde{c}\| \cdot \eta_i.
\end{equation}
Thus the gap between hub cosine and expected pairwise cosine is
\begin{equation}
\Delta_i := \eta_i - \mathbb{E}_J[\cos(\tilde{x}_i,\tilde{x}_J)]
= (1-\|\tilde{c}\|)\,\eta_i.
\end{equation}

\noindent\textbf{2. Beat-count expectation.}
Define $B_i := \#\{j \neq i : \cos(\tilde{x}_i,\tilde{x}_j) > \eta_i\}$, the number of
stored points beating the hub. By linearity of expectation,
\begin{equation}
\mathbb{E}[B_i]
= (n-1) \cdot \Pr_J(\cos(\tilde{x}_i,\tilde{x}_J) > \eta_i).
\end{equation}

\noindent\textbf{3. Tail estimate via Gaussian approximation.}
Working in direction space, $\cos(\tilde{x}_i,\tilde{x}_J) = \tilde{x}_i \cdot \tilde{x}_J$ is a
linear functional of $\tilde{x}_J$, so it has mean $\tilde{x}_i\cdot\tilde{c}=\|\tilde{c}\|\eta_i$
and variance $\sigma_i^2 := \tilde{x}_i^T\tilde{\Sigma}\tilde{x}_i$. 
Because the covariance $\tilde{\Sigma}$ is defined on the unit sphere, $\sigma_i^2$ is
dimensionless and matches the scale of the cosine; since $\tilde{x}_i$ is a unit vector,
$\sigma_i^2 \leq \|\tilde{\Sigma}\|_{\mathrm{op}} = \tilde{m}_1/r(\tilde{\Sigma})$.
Using the Gaussian tail estimate,
\begin{equation}
\Pr(\cos(\tilde{x}_i,\tilde{x}_J) > \eta_i)
\lesssim \exp\!\left(-\frac{\Delta_i^2}{2\sigma_i^2}\right)
\leq \exp\!\left(-\frac{r(\tilde{\Sigma})\,\Delta_i^2}{2\tilde{m}_1}\right).
\end{equation}
Define $\gamma := \Delta_i^2/(2\tilde{m}_1)$. Thus,
\begin{equation}
\mathbb{E}[B_i] \lesssim (n-1)\exp(-r(\tilde{\Sigma})\,\gamma).
\end{equation}

\noindent\textbf{4. From expectation to fraction.}
Let $S := \sum_{i=1}^n B_i$, so $\mathbb{E}[S] \lesssim n(n-1)\exp(-r(\tilde{\Sigma})\gamma)$.
Define $\mathrm{Bad} := \{i : B_i \geq 1\}$. Since $S \geq |\mathrm{Bad}|$, Markov's inequality
gives, for any $\varepsilon \in (0,1)$,
\begin{equation}
\Pr(|\mathrm{Bad}| \geq \varepsilon n)
\leq \frac{\mathbb{E}[S]}{\varepsilon n}
\lesssim \frac{n \exp(-r(\tilde{\Sigma})\gamma)}{\varepsilon}.
\end{equation}
For this to fall below a prescribed $\delta \in (0,1)$, it suffices that
$r(\tilde{\Sigma})\,\gamma \geq \log(n/(\varepsilon\delta))$.
Substituting $\gamma = (1-\|\tilde{c}\|)^2\eta_i^2/(2\tilde{m}_1)$ yields
\begin{equation}
r(\tilde{\Sigma}) \;\gtrsim\; \frac{2\tilde{m}_1}{(1-\|\tilde{c}\|)^2\eta_i^2}\,
                    \log\frac{n}{\varepsilon\delta}.
\end{equation}
When this holds, $\Pr(|\mathrm{Bad}| \geq \varepsilon n) \lesssim \delta$, i.e.\ the hub is the
nearest neighbor of at least $(1-\varepsilon)n$ points with probability at least $1-\delta$.

\medskip
\noindent\textbf{Interpretation.}
The condition mirrors the Euclidean case and exposes the same two forces: a larger corpus
($n\uparrow$) raises $\log(n/\varepsilon\delta)$ and makes the hub harder to sustain, while
stronger anisotropy ($r(\tilde{\Sigma})\downarrow$) inflates the cosine variance and lets more
points overtake the hub. The bound is
informative only when the centrality gap $\Delta = (1-\|\tilde{c}\|)\eta_i$ is bounded away from $0$;
for the embedding models we study this holds, and our experiments confirm the predicted scaling (Table~\ref{tab:centrality_gap}).

\begin{table}[htbp]
\centering
\small
\setlength{\tabcolsep}{10pt}
\caption{Centrality gap $\Delta = (1-\|\tilde{c}\|)\eta$ on HotpotQA~\cite{yang2018hotpotqa} with model BGE~\cite{xiao2024c}, Contriever~\cite{Izacard2021UnsupervisedDI}, and GTE~\cite{Zhang2024mGTEGL}.}
\label{tab:centrality_gap}
\begin{tabular}{l c c c}
\toprule
& \textbf{BGE} & \textbf{Contriever} & \textbf{GTE} \\
\midrule
$\Delta$  & 0.23 & 0.24 & 0.24 \\
\bottomrule
\end{tabular}
\end{table}

\paragraph{Discussion: From Single Gaussian to Gaussian Mixtures.}
The theoretical analysis above adopts a single anisotropic Gaussian model to obtain clean, interpretable conditions. Real-world textual embeddings, however, are notoriously multi-modal: documents spanning diverse topics form distinct semantic clusters in the embedding space.
Under such a mixture distribution, the global centroid may not lie within any single semantic region, which weakens centrality-driven hubness.
This multi-modality can be addressed by partitioning the embedding space into sufficiently many small clusters.
Intuitively, fine-grained clustering isolates different semantic modes from one another, so that within each cluster the local distribution is approximately unimodal, bringing it closer to the single-Gaussian regime assumed by Theorem~5.3.
Furthermore, since Theorem~5.3's sufficient condition scales with \(\log n\), a smaller per-cluster size \(n_j\) substantially weakens the dependence on database size, making the hubness condition easier to satisfy locally.

\paragraph{Validation.}

We validate the centrality-driven hubness phenomenon on real embeddings. We construct a corpus from HotpotQA~\cite{yang2018hotpotqa} and embed the documents using three BGE models~\cite{xiao2024c} with dimensions 384, 768, and 1024. Let \(X = \{x_1, \ldots, x_n\}\) denote the database vectors with centroid \(c\), and \(Q = \{q_1, \ldots, q_m\}\) the query vectors. For each vector we compute its \emph{hubness probability}, i.e., the probability that the centroid is closer than any other database vector: \(\Pr_X[d(x_i, c) < \min_{j \neq i} d(x_i, x_j)]\) for corpus vectors and \(\Pr_Q[d(q_i, c) < \min_{j} d(q_i, x_j)]\) for query vectors. Figure~\ref{fig:hubness_heatmaps} reports this probability as a heatmap, where rows correspond to embedding dimensions and columns to database sizes ranging from 100 to 100K. We evaluate under two distance metrics (Euclidean and Cosine) and two centroid scopes: the \emph{global} centroid of the entire database and \emph{cluster-wise} centroids obtained by partitioning the database with \(k\)-means (each cluster containing about 100 vectors). The upper and lower halves of the figure correspond to corpus and query vectors, respectively.
The results reveal pervasive centrality-driven hubness. Under the global centroid, hubness probability is high for small databases but decays as the database grows, since a denser space weakens the centroid's dominance. Cluster-wise centroids dramatically amplify the effect: hubness probability exceeds 0.84 under Euclidean and 0.74 under Cosine across all database sizes, because each local centroid only needs to dominate a small neighborhood, which directly explains the effectiveness of the cluster-wise Black-Hole Attack. Crucially, the lower half of Figure~\ref{fig:hubness_heatmaps} shows that \textbf{unseen queries} exhibit similar patterns to corpus vectors, confirming that centroids attract independently drawn queries as well.

To bridge the gap between theory and real data, we conduct an experiment on ideal Gaussian distributions. We sweep two modality settings: \emph{single-mode}, where the entire corpus is drawn from one Gaussian (a single semantic cluster), and \emph{multi-mode}, where the corpus is drawn from a mixture of Gaussians with 100 points per mode (emulating the many semantic clusters of real embeddings).
Figure~\ref{fig:ideal_hubness} reports the results as four heatmaps. The experiment confirms our theory along two axes. First, dimension governs whether hubness arises at all: at $d{=}10$ the space is too low-dimensional and Pr stays near $0$, but once $d{\ge}100$ the injected centroid dominates, with Pr reaching $1.0$ at $d{=}1000$. Second, in the single-mode setting Pr decays as the corpus grows. In the multi-mode setting, the per-mode size is fixed at 100 points regardless of $n$, which removes the corpus-size dependence: Pr no longer decays with $n$ and instead stabilizes around $0.6$--$0.7$, depending only on dimension. This directly mirrors our cluster-wise design, where partitioning the space into many small clusters keeps each local distribution effectively low-$n$ and unimodal, so the hubness condition is easily satisfied within every cluster.

\begin{figure}[ht]
  \centering
  \includegraphics[width=\linewidth]{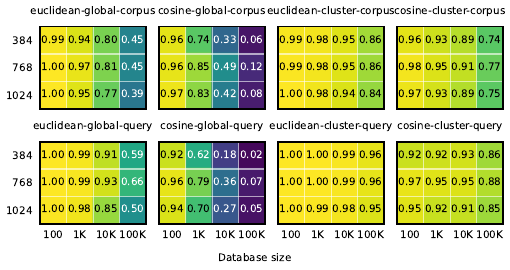}
  \caption{Hubness probability on real embeddings. Rows correspond to embedding dimensions (384, 768, 1024) and columns to database sizes (100--100K). Each cell value is Pr, the fraction of vectors for which the centroid is the nearest neighbor.}
  \label{fig:hubness_heatmaps}
\end{figure}

\begin{figure}[ht]
    \centering
    \includegraphics[width=\linewidth]{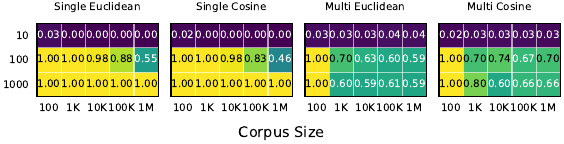}
    \caption{Hubness probability on ideal Gaussian distributions. \emph{Single} denotes a unimodal corpus drawn from one isotropic Gaussian; \emph{Multi} denotes a multi-modal corpus drawn from a mixture of Gaussians.}
    \label{fig:ideal_hubness}
\end{figure}

\section[Evaluation on Attacks]{Evaluation on Attacks\footnote{Due to space constraints, additional experiments are provided in the appendix, including: (i)~sensitivity to the poisoning rate, (ii)~an analysis of cluster overlap, (iii)~the full cluster-size imbalance breakdown, (iv)~attack effectiveness on non-Transformer-generated embeddings such as SIFT~\cite{Jgou2011ProductQF}, DEEP~\cite{babenko2016efficient}, and SPACEV~\cite{simhadri2022results}, and (v)~the effect of mixing vectors from different embedding models.}}
\label{evaluation}

\subsection{Evaluation Setup}
\subsubsection{Model Settings}
We adopt the following embedding models to generate vectors: Contriever~\cite{Izacard2021UnsupervisedDI}, a strong dense retriever trained with contrastive learning; BGE-base-en-v1.5~\cite{xiao2024c}, BAAI's general-purpose embedding model; and GTE-base-en-v1.5~\cite{Zhang2024mGTEGL}, an English text embedding model based on the Generalized Text Embedding architecture.

\subsubsection{Dataset Preparation}
Our evaluation involves 14 publicly available text datasets: Natural Questions (NQ)~\cite{kwiatkowski2019natural},
HotpotQA~\cite{yang2018hotpotqa}, MSMARCO~\cite{Campos2016MSMA}, Climate-FEVER~\cite{diggelmann2020climatefever},
NFCorpus~\cite{boteva2016nfcorpus}, DBPedia~\cite{lehmann2015dbpedia}, FiQA-2018~\cite{maia2018fiqa},
FEVER~\cite{thorne2018fever}, SciFact~\cite{wadden2020fact}, Touch\'{e}-2020~\cite{bondarenko2020touche},
SCIDOCS~\cite{cohan-etal-2020-specter}, Quora~\cite{quora2017}, TREC-COVID~\cite{voorhees2020treccovid}, and
ArguAna~\cite{wachsmuth2018arguana}. All 14 datasets are from the BEIR benchmark~\cite{thakur2021beir}.
We also include three pre-embedded datasets:
DBpedia-OpenAI-1M~\cite{supabase2024dbpedia},
Emoji-GTE~\cite{pszemraj2024emoji},
and Wiki-CS-E5~\cite{karmiq2024wiki}.
For each text dataset, we use its original corpus and queries.
Due to the computational cost of brute-force search, for each evaluation run we randomly sample 3,000 queries
from the query set to measure attack effectiveness.

\subsubsection{Vector Database Settings}
We employ FAISS~\cite{douze2024faiss} as our vector database backend and evaluate our approach using four distinct index types: (1) Flat (exact brute-force search), (2) graph-based HNSW~\cite{Malkov2020Efficient}, (3) clustering-based IVF-Flat~\cite{Jgou2011ProductQF}, and (4) IVF-PQ~\cite{Jgou2011ProductQF}. Since the Flat index performs exhaustive search, it serves as our ground truth for nearest neighbors. Unless specified otherwise, \textbf{cosine similarity} is used as the distance metric.

\subsubsection{Metrics}


We evaluate using four metrics: Recall@K (R@K), Malicious Occupancy@K (MO@K), First Poisoned Rank (FPR), and Attack Success Rate (ASR).

Let $\mathcal{D}=\{x_i\}_{i=1}^{N}$ denote the benign database vectors and $\mathcal{P}=\{p_j\}_{j=1}^{M}$ the injected malicious vectors. The poisoned database is $\mathcal{D}'=\mathcal{D}\cup\mathcal{P}$. Let $\mathcal{Q}$ denote the set of evaluation queries. For each query $q\in\mathcal{Q}$, the vector database returns the ranked top-$K$ results $R_K(q)=[r_1(q),\ldots,r_K(q)]$.


\paragraph{Malicious Occupancy@\emph{K} (MO@\emph{K})}
MO@K quantifies the fraction of the top-$K$ results occupied by malicious items in the poisoned database. Define the number of malicious results in the top-$K$ list as
\[
 m_K(q)=\sum_{t=1}^{K}\mathbf{1}\big[r_t(q)\in \mathcal{P}\big],
\]
then
\[
\mathrm{MO@}K=\frac{1}{|\mathcal{Q}|}\sum_{q\in\mathcal{Q}}\frac{m_K(q)}{K}.
\]

\paragraph{First Poisoned Rank (FPR)}
FPR measures how early the first malicious item appears within the top-$K$ list:
\[
\mathrm{FPR}(q)=\min\{t\in\{1,\ldots,K\}: r_t(q)\in\mathcal{P}\},\quad \min\emptyset:=0.
\]
We report the mean FPR over queries with $\mathrm{FPR}(q)>0$.

\paragraph{Attack Success Rate (ASR)}
ASR is the fraction of queries for which at least one malicious item appears in the top-$K$ results:
\[
\mathrm{ASR}=\frac{1}{|\mathcal{Q}|}\sum_{q\in\mathcal{Q}}\mathbf{1}\big[\mathrm{FPR}(q)>0\big].
\]

\subsection{Overall Attack Effectiveness}

In this section, we report the overall effectiveness of the proposed Black-Hole Attack.
We evaluate three attack scenarios: the cluster-wise attack using full database access, the Surrogate Dataset
Attack, and the attack on pre-embedded datasets.

\paragraph{\textbf{Full Database Access Attack.}}
We inject malicious vectors at a poisoning rate of 1\%, cluster the benign database into more than 3000 clusters, and
insert malicious vectors into the centroid of each cluster.
Table~\ref{tab:malicious_exposure} reports the results across three models, three datasets, and four index types.
The attack is universally effective: MO@10 reaches 68\%--81\% on HotpotQA (ASR~$>$~90\%), 47\%--59\% on NQ, and 29\%--58\% on MSMARCO, with FPR values as low as 2.75.
All three ANN indexes---HNSW, IVF-Flat, and IVF-PQ---exhibit attack effectiveness comparable to the exact Flat index, offering no inherent resilience.
The vulnerability is model-agnostic: while Contriever is the most affected, the attack remains highly effective on BGE and GTE as well.
\begin{table}[ht]
\footnotesize
\setlength{\tabcolsep}{3.5pt}
\renewcommand{\arraystretch}{0.95}
\caption{Full Database Access Attack Effectiveness}
\label{tab:malicious_exposure}
\centering
\begin{threeparttable}
\begin{tabular}{l l l S[table-format=2.2] S[table-format=2.2] S[table-format=2.2] S[table-format=2.2]}
\toprule
\textbf{Dataset} & \textbf{Model} & \textbf{Metric} & \textbf{Flat} & \textbf{HNSW} & \textbf{IVF-F} & \textbf{IVF-PQ} \\
\midrule
\multirow{9}{*}{HotpotQA}
& \multirow{3}{*}{Contriever}
  & MO@10 & 76.97 & 80.86 & 79.98 & 75.36 \\
& & ASR   & 93.43 & 94.07 & 95.23 & 93.47 \\
& & FPR   &  3.10 &  2.75 &  2.83 &  2.98 \\
\cmidrule(lr){2-7}
& \multirow{3}{*}{BGE}
  & MO@10 & 72.71 & 77.66 & 73.79 & 67.58 \\
& & ASR   & 94.40 & 95.30 & 95.10 & 91.13 \\
& & FPR   &  3.58 &  3.11 &  3.48 &  3.84 \\
\cmidrule(lr){2-7}
& \multirow{3}{*}{GTE}
  & MO@10 & 68.10 & 72.13 & 70.01 & 65.09 \\
& & ASR   & 90.43 & 91.83 & 92.13 & 87.43 \\
& & FPR   &  3.99 &  3.60 &  3.81 &  4.07 \\
\midrule
\multirow{9}{*}{MSMARCO}
& \multirow{3}{*}{Contriever}
  & MO@10 & 54.87 & 58.36 & 57.00 & 55.28 \\
& & ASR   & 72.10 & 75.13 & 74.87 & 73.27 \\
& & FPR   &  5.12 &  4.81 &  4.94 &  5.02 \\
\cmidrule(lr){2-7}
& \multirow{3}{*}{BGE}
  & MO@10 & 29.62 & 31.03 & 30.84 & 30.14 \\
& & ASR   & 46.23 & 47.53 & 48.00 & 47.67 \\
& & FPR   &  7.42 &  7.29 &  7.31 &  7.35 \\
\cmidrule(lr){2-7}
& \multirow{3}{*}{GTE}
  & MO@10 & 30.54 & 31.23 & 31.95 & 30.64 \\
& & ASR   & 48.10 & 48.83 & 50.07 & 48.53 \\
& & FPR   &  7.36 &  7.30 &  7.23 &  7.29 \\
\midrule
\multirow{9}{*}{NQ}
& \multirow{3}{*}{Contriever}
  & MO@10 & 52.91 & 55.32 & 54.44 & 51.67 \\
& & ASR   & 72.63 & 73.87 & 74.50 & 72.47 \\
& & FPR   &  5.31 &  5.08 &  5.18 &  5.32 \\
\cmidrule(lr){2-7}
& \multirow{3}{*}{BGE}
  & MO@10 & 57.90 & 59.41 & 59.02 & 57.96 \\
& & ASR   & 80.30 & 80.97 & 81.43 & 80.67 \\
& & FPR   &  4.90 &  4.76 &  4.79 &  4.84 \\
\cmidrule(lr){2-7}
& \multirow{3}{*}{GTE}
  & MO@10 & 47.18 & 48.07 & 49.15 & 46.03 \\
& & ASR   & 69.90 & 70.37 & 72.97 & 69.33 \\
& & FPR   &  5.87 &  5.79 &  5.70 &  5.88 \\
\bottomrule
\end{tabular}

\begin{tablenotes}[flushleft]
\footnotesize
\item Higher MO@10 and ASR indicate a stronger attack, while lower FPR indicates a stronger attack.
\end{tablenotes}
\end{threeparttable}

\end{table}

\paragraph{\textbf{Surrogate Dataset Attack.}}
Figure~\ref{fig:transfer_attack} evaluates the Surrogate Dataset Attack, where the attacker constructs malicious vectors solely from publicly available query sets without any access to the target database's internal embeddings. Query-Corpus similarity is the mean cosine similarity between query and corpus embeddings from the same source, and Query-Query similarity is the mean pairwise cosine similarity among surrogate query embeddings. For comparison, we also plot the Full Database Access Attack (orange dashed line), which uses the victim's own corpus embeddings to construct malicious vectors. The Surrogate Attack (blue solid line) achieves strong effectiveness on 10 of 14 datasets (MO@10 $>$ 30\%), with climate-fever reaching 98.3\%, while the four failures---scidocs, quora, trec-covid, and arguana---all share high Query-Corpus similarity (0.594--0.949), indicating near-overlap between query and corpus distributions that mitigates the performance of attack. Notably, on several datasets (climate-fever, hotpotqa, nfcorpus, dbpedia, fiqa, fever), the Surrogate Attack even surpasses the Full Access baseline, demonstrating that precisely matching the victim's corpus distribution is unnecessary: a surrogate query pool can exploit the same hubness geometry more effectively than the victim's own embeddings. Attack effectiveness depends on this distributional gap rather than on accurately guessing user queries, dramatically lowering the practical barrier of the Black-Hole Attack.

\begin{figure}[ht]
    \centering
    \includegraphics[width=\linewidth]{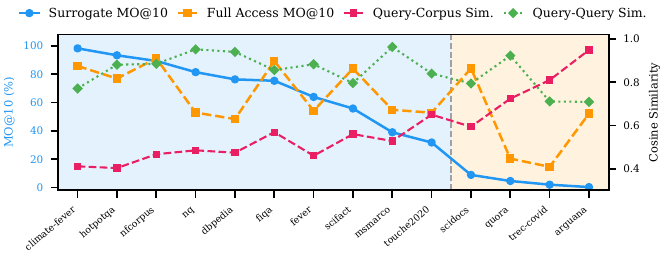}
    \caption{Surrogate Dataset Attack vs.\ Full Database Access Attack across 14 BEIR datasets.}
    \label{fig:transfer_attack}
\end{figure}

\paragraph{\textbf{Poisoning Public Pre-Embedded Datasets.}}

Table~\ref{tab:pure-vector-results} reports the attack on three publicly available pre-embedded
datasets, where the attacker directly downloads the pre-computed vectors, constructs malicious
vectors via clustering, and injects them---without ever running an embedding model.
Across all three datasets, the attack achieves MO@10 above 71\%, ASR above 93\%, and FPR below
4, covering different embedding models (OpenAI text-embedding-3-large, GTE-large-en-v1.5, E5-base)
and diverse data modalities (entity descriptions, emoji-text pairs, Wikipedia articles).
This demonstrates that the Black-Hole Attack readily extends to real-world pre-embedded datasets
shared on open platforms such as Hugging Face.
\begin{table}[htbp]
  \centering
  \small
  \caption{Results on precomputed-vector datasets.}
  \label{tab:pure-vector-results}
  \begin{tabular}{lccc}
      \toprule
      \textbf{Dataset} & \textbf{MO@10} & \textbf{ASR} & \textbf{FPR} \\
      \midrule
      DBpedia OpenAI 1M & $71.10$ & 93.74 & $3.76$ \\
      Emoji GTE & $75.52$ & 98.66 & $3.43$ \\
      Wiki CS E5 & $72.38$ & 94.09 & $3.62$ \\
      \bottomrule
  \end{tabular}
\end{table}

\subsection{Robustness Analysis}

\begin{figure*}[htbp]
    \centering
    \includegraphics[height=5pt]{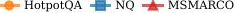}
    \vspace{6pt}

    \begin{subfigure}{0.2\linewidth}
        \centering
        \includegraphics[width=\linewidth]{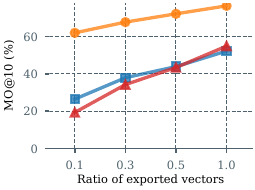}
        \caption{Partial Export}
        \label{fig:part}
    \end{subfigure}
    \hfill
    \begin{subfigure}{0.2\linewidth}
        \centering
        \includegraphics[width=\linewidth]{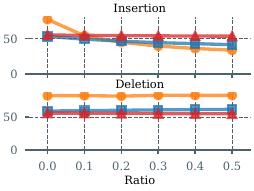}
        \caption{Addition\&Deletion}
        \label{fig:adddel}
    \end{subfigure}
    \hfill
    \begin{subfigure}{0.2\linewidth}
        \centering
        \includegraphics[width=\linewidth]{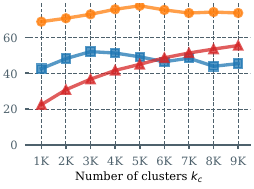}
        \caption{Number of Clusters}
        \label{fig:cluster}
    \end{subfigure}
    \hfill
    \begin{subfigure}{0.2\linewidth}
        \centering
        \includegraphics[width=\linewidth]{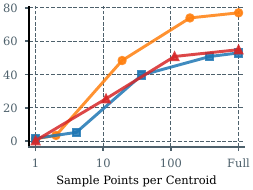}
        \caption{Clustering Quality}
        \label{fig:cluster_quality}
    \end{subfigure}
    \caption{Robustness and Sensitivity Analysis of the Black-Hole Attack.}
    \label{fig:ablation_all}
\end{figure*}

\subsubsection{Robustness to Partial Database Export}
To analyze the impact of partial database access, we evaluate the attack under export ratios ranging from 0.1 to 1.0.
Figure~\ref{fig:part} demonstrates that our attack remains highly effective under limited exports, with MO@10 improving
steadily as the ratio increases. On the MSMARCO dataset, expanding the export ratio from 0.1 to 1.0 improves MO@10 from 19.3\%
to 54.8\%—a 184\% increase—while on NQ, MO@10 nearly doubles from 26.4\% to 52.1\%. Notably, HotpotQA achieves an MO@10 of 61.9\%
at only 10\% export, indicating that even minimal database access suffices for a practical attack. These results highlight
that our attack does not depend on full corpus access and scales predictably with the exported portion.

\subsubsection{Robustness to Dynamic Data Updates}

Real-world vector databases are continually updated, so we evaluate how dynamic data updates affect the attack. We vary both data addition and deletion from 0\% to 50\% of the corpus size. To test the attack under \emph{strong} distributional drift, the inserted vectors are drawn from \emph{other} datasets rather than the target corpus itself. The results (Figure~\ref{fig:adddel}) show that \textbf{deletion} slightly \emph{improves} attack effectiveness, as removing benign vectors thins the neighborhood around the malicious centroids. \textbf{Insertion} shifts the data distribution and degrades the attack, but the effect is limited: even after 50\% insertion, MSMARCO drops by only 1\% in MO@10, and HotpotQA still maintains 33\% MO@10.
In addition to data churn, production systems also rebuild their indexes and occasionally re-embed their corpora. \textbf{Re-indexing} does not affect the attack: as Table~\ref{tab:malicious_exposure} shows, attack effectiveness remains consistent across four index types, since the attack exploits embedding geometry rather than index structure. \textbf{Re-embedding} the corpus with a new model removes the injected malicious vectors from the database entirely, and is therefore an effective defense. However, its prohibitive computational cost makes it impractical for frequent deployment in production systems.

\subsection{Sensitivity Analysis}



\subsubsection{Sensitivity to Number of Clusters}
To analyze the effect of the cluster-wise design in Section~\ref{cluster_wise}, we vary the number of cluster centers used to construct local black-hole centers. Figure~\ref{fig:cluster} reports the resulting MO@10 as a function of \(L\) under a fixed poisoning rate of 1\%. We observe that MO@10 generally increases as \(L\) grows, indicating stronger malicious presence in the Top-10 results. It suggests that a moderate number of cluster centers is sufficient to cover most query regions for this dataset.

\subsubsection{Sensitivity to Clustering Quality}
We use FAISS-GPU to cluster the corpus, which proceeds in two stages: (1) it \emph{trains} the centroids on a sample of the corpus, and (2) it \emph{assigns} all corpus vectors
to the learned centroids. The training sample size governs the quality of clustering. We fix the number of centroids $L$ and the number of training iterations to 25, then vary the number of training samples per centroid $s$
(so the total training set size is $L \cdot s$), thereby controlling clustering quality.
Figure~\ref{fig:cluster_quality} reports MO@10 under increasing $s$.
When only one point per centroid is sampled ($s{=}1$), MO@10 remains below 2\% across all datasets.
As $s$ grows, MO@10 rises sharply and saturates under full sampling at 76.97\%, 52.91\%, and 54.87\% on HotpotQA,
NQ, and MSMARCO, respectively.
Overall, higher clustering quality consistently leads to stronger attack performance, confirming that a well-clustered
corpus representation is essential to the effectiveness of our method.

\subsubsection{Sensitivity to Cluster Size}
\label{subsec:cluster_size}
We analyze how cluster size affects the Black-Hole Attack by measuring MO@10 across clusters of different sizes. As shown in Figure~\ref{fig:cluster_size_perf},
$k$-means clustering produces clusters whose sizes span a wide range, from around 100 up to 2K+. The results show that cluster size has different effects across datasets, and MO@10 remains consistently significant across all size groups. This confirms that the Black-Hole Attack is robust to the size imbalance of
$k$-means clustering and does not rely on uniformly sized clusters.


\subsubsection{Sensitivity to Cluster Density}
We further analyze how the local density around a query affects the Black-Hole Attack. For each query, we define its \emph{local density} as the mean similarity to its 10 nearest
legitimate documents: a higher value means the query sits in a denser region. We sort queries by this density and
measure MO@10 across density groups, from the sparsest to the densest. As shown in Figure~\ref{fig:cluster_size_perf}, MO@10 decreases as density increases: when
a query is close to corpus vectors, those nearby documents compete more strongly with the injected centroid and pull down the attack effectiveness. Nevertheless, even in the
densest regions MO@10 remains significant across all datasets, showing that higher density only weakens the attack gradually rather than eliminating it.

\begin{figure}[htbp]
    \centering
    \includegraphics[width=\linewidth]{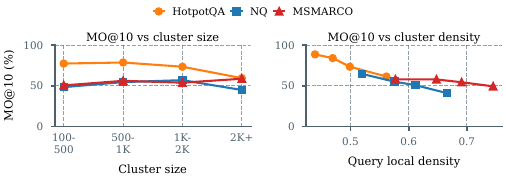}
    \caption{MO@10 versus cluster size and cluster density. Query local density is defined as the mean similarity of a query to its 10 nearest legitimate documents; a higher value indicates a denser region.}
    \label{fig:cluster_size_perf}
\end{figure}

\subsubsection{Sensitivity to ANN Index Parameters}
\label{subsec:ann_parameter}
We systematically ablate the key parameters of three widely used ANN indexes---HNSW, IVF, and IVFPQ, comparing to a Flat baseline that uses exact search: 77.30\%. The results are consolidated in Figure~\ref{fig:ann_parameters}.
Across all three index types, attack effectiveness fluctuates only moderately and remains close to the Flat baseline under all practical parameter settings.
For HNSW, varying ef\_search (16--256), M (4--64), and ef\_construction (40--800) yields MO@10 consistently within $\pm$10\% of Flat.
For IVF, varying nprobe (16--256) and nlist (1024--8192) produces similarly stable results.
For IVFPQ, setting the sub-quantizer count M to 48 slightly weakens the attack (still around 50\%), but this extreme compression already degrades retrieval accuracy; at standard settings ($M\ge 64$), the attack remains intact.
Together, these results confirm that under different ANN index parameters, the Black-Hole Attack is always effective.

\begin{figure}[htbp]
    \centering
    \includegraphics[width=\linewidth]{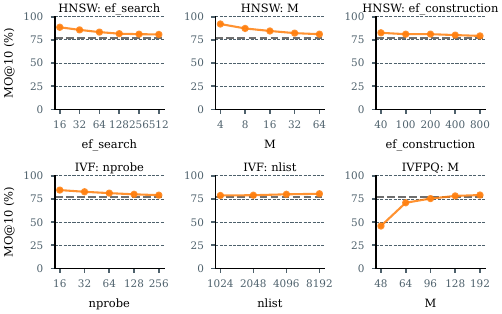}
    \caption{Sensitivity of the Black-Hole Attack to ANN index parameters (HNSW, IVF, IVFPQ) on HotpotQA with Contriever. Gray dashed lines indicate the Flat (exact search) baseline.}
    \label{fig:ann_parameters}
\end{figure}

\subsection{Cost and Analysis}

\subsubsection{Scalability to Larger Corpora}
\label{subsec:scalability_msmarco_v2}

To verify that the Black-Hole Attack scales beyond the 1M--10M range, we conduct experiments on MS~MARCO~v2~\cite{Campos2016MSMA,enevoldsen2025mmteb,muennighoff2022mteb} (138M documents, Contriever embeddings). We evaluate both the Full Database Access Attack and the Surrogate Dataset Attack. To obtain a sufficiently large surrogate query set, we leverage the query set of MS~MARCO~Web~Search~\cite{chen2024msmarco}, whose queries are independently collected and have no overlap with the MS~MARCO~v2 queries, eliminating any risk of query leakage. We filter 4M English queries from this dataset and combine them with queries from other BEIR datasets to form the surrogate pool.
Table~\ref{tab:msmarco_v2} reports the results. Both attack variants achieve strong effectiveness on this 138M-document corpus. The Surrogate Dataset Attack attains MO@10 of 81.53\% and ASR of 89.80\%, and the Full Database Access Attack reaches MO@10 of 63.99\% and ASR of 75.60\%. Notably, both methods perform even better than on the 9M MS~MARCO benchmark, confirming that the Black-Hole Attack remains effective at larger database sizes.

\begin{table}[htbp]
    \centering
    \small
    \setlength{\tabcolsep}{5pt}
    \caption{Black-Hole Attack on MS~MARCO~v2 (138M documents, Contriever).}
    \label{tab:msmarco_v2}
    \begin{tabular}{l c c c}
    \toprule
    \textbf{Attack Scenario} & \textbf{MO@10} & \textbf{ASR} & \textbf{FPR} \\
    \midrule
    Full Database Access     & 63.99 & 75.60 & 4.29 \\
    Surrogate Dataset Attack & 81.53 & 89.80 & 2.70 \\
    \bottomrule
    \end{tabular}
    \end{table}

\subsubsection{Analysis to Clustering Cost}
\label{subsec:cluster_cost}
We examine the cost of the clustering using MSMARCO-v2. We use the FAISS-GPU $k$-means implementation on an H200 GPU and measure the clustering time along two axes: (i) the corpus size, varied from 100K to 100M vectors at a fixed sampling budget (25\% of the corpus); and (ii) the number of training samples per centroid, varied on a fixed 10M-vector corpus. The results are shown in Figure~\ref{fig:cluster_scale}. Even at 100M vectors, clustering completes within three hours. For corpora at the billion-vector scale, the attacker can alternatively use the Surrogate Dataset Attack: it clusters a far smaller set of surrogate queries rather than the entire corpus, sidesteps the large-corpus clustering cost, and still achieves strong attack effectiveness.

\begin{figure}[htbp]
    \centering
    \includegraphics[width=\linewidth]{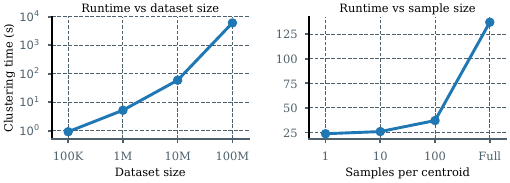}
    \caption{Clustering cost on an H200 GPU. Left: runtime versus corpus size at a fixed sampling budget. Right: runtime versus the number of training samples per centroid on a fixed 10M-vector corpus.}
    \label{fig:cluster_scale}
\end{figure}

\subsection{Analysis on Downstream RAG Applications}

\subsubsection{Attack Effectiveness on QA Pipelines}
We use the reproduced LongRAG~\cite{zhao-etal-2024-longrag} pipeline as the downstream application and poison its underlying vector database with the Black-Hole Attack. Figure~\ref{fig:ablation_rerank} reports the end-to-end QA performance before and after poisoning. The results show that the attack substantially degrades answer quality on all three benchmarks. On HotpotQA, the F1 score drops from 51.8 to 37.0; on 2WikiMQA~\cite{ho-etal-2020-constructing}, it drops from 44.5 to 33.3; and on MuSiQue~\cite{trivedi-etal-2022-musique}, it drops from 23.2 to 7.5. These results indicate that poisoning only the vector database is sufficient to severely impair the performance of a strong downstream RAG system.
We further evaluate the effect of varying the rerank depth. Figure~\ref{fig:ablation_rerank} shows that while more candidates improve the clean F1 as expected, the attack success rate remains above 80\% across all datasets even under the largest rerank budget tested. This indicates that reranking alone, without dedicated defense mechanisms, does not meaningfully mitigate the Black-Hole Attack.

\begin{figure}[htbp]
    \centering
    \includegraphics[width=\linewidth]{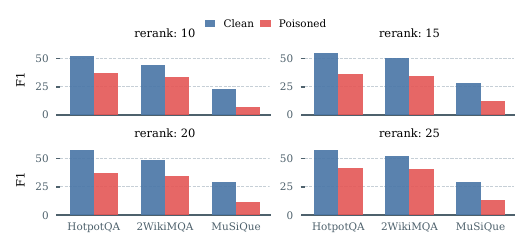}
    \caption{Impact of the Black-Hole Attack on downstream RAG applications.}
    \label{fig:ablation_rerank}
\end{figure}

\subsubsection{Impact of Filtering-Based Defenses on the Black-Hole Attack}
RAG pipelines often employ clustering-based methods for \textbf{semantic filtering} and \textbf{duplicate detection} to defend against adversarial poisoning. SeCon-RAG~\cite{si2026secon} partitions the corpus into clusters and flags vectors too close to cluster centroids as malicious---a design that mirrors the Full Database Access Attack. Figure~\ref{fig:secon_rag} reports the confusion matrices under both attack modes on HotpotQA. At 10,000 clusters, the filter flags 98.3\% of malicious vectors under Full Database Access (TP), but only 13.5\% of malicious vectors are flagged under Surrogate Dataset Attack. This gap arises because the surrogate query distribution diverges from the victim corpus distribution; centroids learned from surrogate queries do not precisely align with the victim's embedding clusters, causing the filter to miss most injected vectors.
TrustRAG~\cite{zhou2025trustrag} applies K=2 clustering to the top-10 retrieval results of each query and flags the denser cluster as malicious. Figure~\ref{fig:trust_rag} shows that this method can flag a substantial fraction of malicious vectors (TP 73.7\% for Full Access, 58.3\% for Surrogate). However, the Black-Hole Attack already dominates the top-10 results---MO@10 is up to 76.97\%. Even after TrustRAG removes most flagged malicious vectors, the remaining top-10 consists of the unflagged malicious ones, with only 1--3 clean documents surviving. Moreover, unlike SeCon-RAG, TrustRAG cannot directly remove malicious vectors from the corpus---it only filters at query time, leaving the poisoned vectors in the database.

\begin{figure}[htbp]
    \centering
    \includegraphics[]{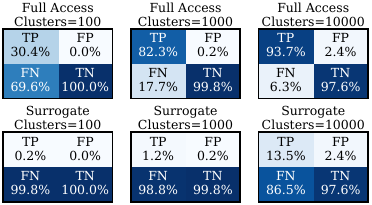}
    \caption{Detection accuracy of SeCon-RAG on malicious vectors. Columns: number of clusters used by the SeCon-RAG filter. TP: malicious vectors correctly flagged. FP: clean vectors falsely removed. FN: malicious vectors missed. TN: clean vectors correctly retained.}
    \label{fig:secon_rag}
\end{figure}

\begin{figure}[htbp]
    \centering
    \includegraphics[width=0.4\linewidth]{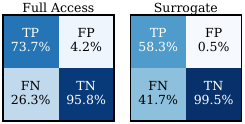}
    \caption{Detection accuracy of Trust-RAG.}
    \label{fig:trust_rag}
\end{figure}

\subsubsection{Robustness to Content-Level Inspection}
Beyond clustering-based filters, RAG pipelines also deploy content-level inspection, including \textbf{source validation} and \textbf{lightweight content sanity checks}. Such defenses, however, are ill-suited to the Black-Hole Attack. Unlike adversarial-corpus-poisoning attacks that perturb a given text to derive an adversarial embedding, the Black-Hole Attack constructs malicious vectors \emph{directly} in the embedding space and only later assigns text content to them. The attacker therefore retains complete freedom to design the accompanying text to evade text-level filters. For source validation, RAGParadox~\cite{choi2025ragparadox} has shown that malicious texts can be paired with fabricated but plausible source attributions, effectively bypassing source-level checks. For lightweight content sanity checks, CPA-RAG~\cite{li2025cparag} shows that adversarially generated malicious texts with strong fluency and generalization can reliably evade perplexity and readability filters.

\section{Defenses Against Black-Hole Attacks}

\paragraph{Hubness Mitigation}
Since the Black-Hole Attack exploits centrality-driven hubness, a natural defense is to reduce hubness in the embedding space. We evaluate three representative methods: ZN (Z-score normalization)~\cite{Fei2021ZScoreNH,Trosten2023HubsAH}, CL2 (Centered L2 normalization)~\cite{Wang2019SimpleShotRN}, and noHub~\cite{Trosten2023HubsAH}.
Table~\ref{tab:defense_merged} reports MO@10 and R@10 under each defense with Contriever embeddings and Flat search. ZN preserves high retrieval quality (R@10 85.83\%--89.67\%) but fails to suppress the attack (MO@10 65.36\%--83.13\%). noHub achieves strong attack suppression on HotpotQA and MSMARCO (MO@10 0.71\%--13.91\%), but degrades R@10 to near zero and leaks severely on NQ (MO@10 47.27\%). CL2 provides the best overall tradeoff, reducing MO@10 to 1.66\%--12.25\% while retaining moderate R@10 (58.65\%--76.70\%). These results show that hubness mitigation can defend against the Black-Hole Attack, but every method forces a difficult choice between security and retrieval utility---none achieves both simultaneously.

\begin{table}[htbp]
\centering
\small
\setlength{\tabcolsep}{5pt}
\caption{Defense effectiveness against the Black-Hole Attack: MO@10 and R@10 with Flat search.}
\label{tab:defense_merged}
\begin{tabular}{l c c c c c c}
\toprule
& \multicolumn{3}{c}{\textbf{MO@10}} & \multicolumn{3}{c}{\textbf{R@10}} \\
\cmidrule(lr){2-4} \cmidrule(lr){5-7}
\textbf{Dataset} & \textbf{ZN} & \textbf{CL2} & \textbf{noHub} & \textbf{ZN} & \textbf{CL2} & \textbf{noHub} \\
\midrule
HotpotQA & 83.13 & 1.66  & 0.71  & 85.83 & 58.65 & 5.45  \\
MSMARCO  & 65.36 & 12.25 & 13.91 & 89.67 & 76.70 & 11.38 \\
NQ       & 68.66 & 10.00 & 47.27 & 87.48 & 66.78 & 10.44 \\
\bottomrule
\end{tabular}
\end{table}

\paragraph{Detection-Based Defense}
We next study a more conservative defense that operates directly on the stored corpus vectors. The idea is to use a small probe set drawn from the database itself, measure how often each stored vector is retrieved as a nearest neighbor of these probes, and remove the vectors that are selected unusually often.
Let the stored corpus vectors be $X=\{x_i\}_{i=1}^{N}\subset\mathbb{R}^{d}$. We first partition $X$ into $L$ clusters and, to cover different semantic regions, sample a small
probe set from each cluster (0.1\% of its vectors, at least one) to form the probe set $P$. For each probe we retrieve its top-$k$ nearest neighbors in $X$, and define the
\emph{hit count} $h_i$ of each vector $x_i$ as the number of probes that retrieve it. Since most vectors are never retrieved ($h_i=0$), we compute the median $m$ over only the
strictly positive hit counts, and mark a vector as suspicious and remove it before retrieval if $h_i>2m$.
When using a brute-force (Flat) index, the detection procedure requires an additional $k$-NN pass for each probe vector, yielding $O(|P|\cdot N\cdot d)$ complexity, i.e., $O(N^2 d)$ in the worst case.
We evaluate this detection-based defense with a probe ratio of 0.1\% per cluster. Table~\ref{tab:defense_surro_full} reports the results under both the Full-Export Attack and the Surrogate Dataset Attack. Under the full-export attack, the defense is highly effective: it reduces MO@10 from 76.97\%, 54.87\%, and 52.91\% to 7.09\%, 1.50\%, and 1.63\% on HotpotQA, MSMARCO, and NQ respectively, while preserving high retrieval quality (R@10 94.51\%--96.34\%). Under the surrogate attack, however, the defense is substantially weakened---on HotpotQA and NQ, MO@10 drops only marginally (from 93.36\% to 91.90\% and from 81.45\% to 68.13\%, respectively). This gap arises from the same distributional mismatch discussed in Section~\ref{impact_filtering_based}: the probe vectors are drawn from the corpus distribution, whereas the surrogate attack places malicious vectors in query-offset regions, so they are rarely retrieved by the probes and evade the hit-count filter. Despite its limitations against surrogate attacks, the detection-based defense is computationally lightweight: with ANN acceleration, the probe search completes in seconds, making it practical to deploy as a routine safeguard against full-export-style threats.

\begin{table}[htbp]
    \centering
    \footnotesize
    \setlength{\tabcolsep}{3pt}
    \caption{Detection-based defense under Full-Export Attack and Surrogate Dataset Attack, and probe search time across ANN index types.}
    \label{tab:defense_surro_full}
    \begin{tabular}{l @{\hspace{2pt}} | @{\hspace{2pt}} c @{\hspace{2pt}} | @{\hspace{2pt}} c c @{\hspace{2pt}} | @{\hspace{1pt}} c c @{\hspace{1pt}} | @{\hspace{1pt}} c c c c}
    \toprule
    \multirow{2}{*}{\textbf{Dataset}} & \multirow{2}{*}{\textbf{R@10}} & \multicolumn{2}{c}{\textbf{Full(MO@10)}} & \multicolumn{2}{c}{\textbf{Surrogate}} & \multicolumn{4}{c}{\textbf{Time (s)}} \\
    & & \textbf{Before} & \textbf{After} & \textbf{Before} & \textbf{After} & \textbf{Flat} & \textbf{HNSW} & \textbf{IVF} & \textbf{IVFPQ} \\
    \midrule
    HotpotQA & 96.34 & 76.97 & 7.09  & 93.36 & 91.90 & 9k  & 4.4 & 7.5 & 1.3 \\
    MSMARCO  & 94.55 & 54.87 & 1.50  & 38.99 & 17.37 & 25k & 6.1 & 16.6 & 2.3 \\
    NQ       & 94.51 & 52.91 & 1.63  & 81.45 & 68.13 & 2k  & 2.1 & 1.8 & 0.6 \\
    \bottomrule
    \end{tabular}
\end{table}



\section{Conclusion}
We present the \emph{Black-Hole Attack}, a query-agnostic poisoning attack against vector databases grounded in \emph{centrality-driven hubness}. By injecting a small number
of malicious vectors at the global or cluster-wise centroid, an attacker can hijack Top-$k$ results for most queries without any knowledge of their content.
We define four attack paths along a capability spectrum and show strong effectiveness across all
scenarios. Experiments across multiple models, datasets, and retrieval backends confirm that a $1\%$ poisoning rate achieves up to $94.4\%$ attack success rate, and that the attack substantially degrades downstream RAG utility. 
Among the defenses we study, hubness mitigation provides only partial protection. The detection-based defense is highly effective against the full-export attack, reducing MO@10 to single digits while preserving retrieval quality. However, it largely fails against the Surrogate Dataset Attack: the probe-based filtering mechanism breaks down because the probe vectors are drawn from the corpus distribution and cannot identify vectors planted in query-offset regions. A robust and adaptive defense remains an open problem.

\balance

\bibliographystyle{ACM-Reference-Format}
\bibliography{ref}

\clearpage
\appendix
\section*{Appendix}

\section{Sensitivity to Poisoning Rate}
\label{app:poisoning_rate}
Table~\ref{tab:ablation_poisoning_rate} reports MO@10 under different poisoning budgets.
Effectiveness increases steadily as the rate grows from 0.1\% to 1.0\%.
On HotpotQA, MO@10 rises from 19.30\% to 93.66\%; similar trends hold for MSMARCO and NQ.
Beyond 1.0\%, further increases yield diminishing returns, with improvements of less than 0.1~percentage~points for MSMARCO and NQ and no increase for HotpotQA.
A poisoning rate around 1\% is therefore sufficient to saturate the attack's impact in this setting.
\begin{table}[htbp]
\centering
\small
\setlength{\tabcolsep}{8pt}
\caption{Sensitivity to Poisoning Rate.}
\label{tab:ablation_poisoning_rate}
\begin{tabular}{l S[table-format=2.2] S[table-format=2.2] S[table-format=2.2] S[table-format=2.2]}
\toprule
\textbf{Dataset} & \textbf{0.1\%} & \textbf{0.5\%} & \textbf{1.0\%} & \textbf{1.5\%} \\
\midrule
HotpotQA   & 41.27 & 69.00 & 76.97 & 77.45 \\
MSMARCO    & 15.17 & 40.65 & 54.87 & 55.23 \\
NQ         & 17.18 & 41.79 & 52.91 & 53.28 \\
\bottomrule
\end{tabular}
\end{table}

\section{Effect of Cluster Overlap}
\label{app:overlap}
Real-world embeddings rarely form compact, well-separated clusters.
The silhouette score $\hat{s}(i)$ is a standard clustering-quality metric: values close to $1$ indicate well-separated clusters and values close to $0$ indicate severe overlap.
Table~\ref{tab:overlap} reports the mean silhouette score on all three datasets.
The mean is approximately $0.06$ in each case, confirming that with a large number of clusters, clean separation is practically unachievable in real embedding spaces.
Despite this severe overlap, the Black-Hole Attack remains effective: MO@10 reaches $76.97\%$, $52.91\%$, and $54.87\%$ on HotpotQA, NQ, and MSMARCO respectively.

\begin{table}[htbp]
\centering
\caption{Cluster overlap vs.\ attack effectiveness. $\bar{s}$ is the mean silhouette score (lower = more overlap).}
\label{tab:overlap}
\small
\begin{tabular}{lcc}
\toprule
Dataset & Mean silhouette $\bar{s}$ & MO@10 \\
\midrule
HotpotQA & 0.061 & 76.97\% \\
NQ       & 0.061 & 52.91\% \\
MSMARCO  & 0.063 & 54.87\% \\
\bottomrule
\end{tabular}
\end{table}

This outcome follows from the centrality-driven hubness mechanism.
Each centroid is computed exclusively from the vectors of its own cluster, so its centrality relative to those vectors is unaffected by overlap with neighboring clusters.
Vectors from other clusters that intrude into the region are likewise attracted toward the same centroid.
Consequently, cluster overlap does not undermine the Black-Hole Attack.

\section{Full Cluster-Size Breakdown}
\label{app:cluster_size}
Due to space constraints, the main paper (Figure~\ref{fig:cluster_size_perf}) reports only MO@10 as a function of cluster size.
Figure~\ref{fig:cluster_size_full} provides the complete breakdown: cluster count, query count, MO@10, and ASR per size group.
Two patterns are consistent across all datasets.
First, both cluster sizes and query distributions are heavily concentrated in the 500--1K range, with similar distributions in the neighboring 100--500 and 1K--2K groups.
Second, MO@10 and ASR remain consistently high across all size groups and do not degrade substantially even for the 2K$+$ category.
Thus, although cluster sizes are imbalanced, the attack is robust to this imbalance.

\begin{figure}[htbp]
    \centering
    \includegraphics[width=\linewidth]{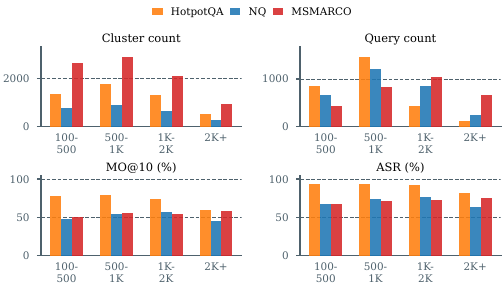}
    \caption{Full cluster-size breakdown: cluster count, query count, MO@10, and ASR per size group.}
    \label{fig:cluster_size_full}
\end{figure}

\section{Applicability to Classical Vector Datasets}
\label{app:classical}
The Black-Hole Attack targets Transformer-generated embeddings~\cite{vaswani2017attention}, which are high-dimensional (typically 384--768d) and exhibit a distinctive distribution.
Classical vector datasets such as SIFT~\cite{Jgou2011ProductQF}, DEEP~\cite{babenko2016efficient}, and SPACEV~\cite{simhadri2022results} are produced by hand-crafted feature extraction or earlier-generation deep learning models rather than by Transformer-based embedding models.
They have low dimensionality (96--128d) and a different data distribution.
Table~\ref{tab:classical_1m} reports the results on these datasets: the hubness phenomenon is observable in each case, but MO@10 remains only $10.97\%$--$26.12\%$, compared to $60\%$--$80\%$ on Transformer embeddings of comparable corpus size.
The severity of the Black-Hole Attack is therefore closely tied to the high-dimensional geometry and distributional characteristics of Transformer embeddings.

\begin{table}[htbp]
    \centering
    \small
    \setlength{\tabcolsep}{5pt}
    \caption{Black-Hole Attack on classical vector datasets.}
    \label{tab:classical_1m}
    \begin{tabular}{l c c c c}
    \toprule
    \textbf{Dataset} & \textbf{Dim} & \textbf{MO@10} & \textbf{ASR} & \textbf{FPR} \\
    \midrule
    SIFT1M   & 128 & 18.10 & 29.07 & 8.39 \\
    DEEP1M   & 96  & 10.97 & 18.30 & 9.04 \\
    SPACEV1M & 100 & 26.12 & 40.97 & 7.71 \\
    \bottomrule
    \end{tabular}
    \end{table}

\section{Cost of Re-embedding}
\label{app:reembedding}
Re-embedding the corpus with a new model would shift the embedding space and invalidate pre-computed malicious vectors.
In principle, this is an effective defense.
In practice, however, re-embedding imposes prohibitive computational cost: every stored record must be passed through the new model.
For MS~MARCO~v2 (138M documents), encoding with 8$\times$H200 GPUs and 8$\times$22 vCPUs exceeds 10~hours.
For organizations using cloud embedding APIs, re-embedding via OpenAI's text-embedding-3-small (\$0.02/1M tokens) and text-embedding-3-large (\$0.13/1M tokens) would cost approximately \$166 and \$1079, respectively, with API rate limits imposing comparable wall-clock latency.
Scaling to billion-document production corpora makes re-embedding a multi-day undertaking on high-end hardware.
Thus, while re-embedding is effective in principle, its cost makes frequent deployment impractical.

\section{Multiple Embedding Models or Versions}
\label{app:multi_model}
Embeddings from different models---or even different versions of the same model---cannot be mixed within a single database instance: the embedding spaces have different dimensionalities, semantic geometries, and distributional properties.
Mixing them severely degrades retrieval quality~\cite{dahlbeck2025compatibility} and is explicitly discouraged in production best practices. Recent work on embedding translation~\cite{maystre2025embedding,yang2026generalizable} has explored methods to reduce the degradation caused by model mixing, but these approaches can only mitigate the loss rather than eliminate it---cross-model retrieval remains substantially worse than within-model retrieval. In our threat model, we therefore assume a single embedding model per database instance, which aligns with standard deployment practice.

To substantiate this, the same set of queries is encoded using v1.5 and v1 of BGE (and separately GTE), retrieved against a v1.5-encoded corpus, and the top-10 overlap between the two query sets is measured.
As shown in Figure~\ref{fig:diff_model}, even a minor training update (BGE) reduces overlap to 48--64\%, while an architectural change (GTE) drops overlap to zero.
Embedding spaces from different model versions are therefore incompatible in practice.

\begin{figure}[htbp]
    \centering
    \includegraphics[width=0.8\linewidth]{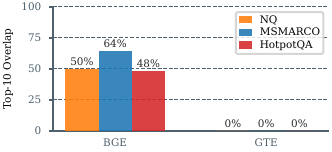}
    \caption{Top-10 retrieval overlap between v1.5 and v1 of the same embedding model family. Even a minor training update (BGE) substantially degrades overlap; an architectural change (GTE) drops it to zero.}
    \label{fig:diff_model}
\end{figure}

\section{Encryption of Stored Embeddings}
\label{app:encryption}
The Black-Hole Attack assumes that stored vectors are accessible in plaintext through the database's normal query API.
A natural defense is to encrypt the stored embeddings.
However, query-level encryption of embeddings is not yet realized in deployed systems.

\textit{Academic approaches.}
VeriANN~\cite{hui2026veriann} replaces distance computation with a frequency-counting algorithm over LSH buckets and distributes index shards across two non-colluding servers, but evaluates accuracy only against frequency-counting baselines rather than standard distance-based ground truth, and its real-world latency and recall remain unvalidated.
CAPRISE~\cite{ye2026pprag} adds random perturbations to embeddings, but the perturbation causes severe ranking drift: to reliably recover top-10 results, the system must retrieve at least 108--571 candidates.
CESSE~\cite{tang2025cesse} applies a linear transformation to embeddings, yet the geometric relationships among vectors remain intact, allowing an attacker to analyze the embedding distribution and launch the Black-Hole Attack directly in the transformed space.
These schemes uniformly trade off retrieval performance and accuracy.

\textit{Industry practice.}
Weaviate~\cite{weaviate2025security} recommends EBS volume encryption (disk-level AES), but vectors remain in plaintext in memory and are accessible through the query API.
Milvus~\cite{milvus2025encryption} provides data-at-rest and data-in-transit protection, while still exposing plaintext vectors through normal API operations.
Qdrant~\cite{qdrant2025security} focuses on access control without mentioning encryption.
A third-party analysis~\cite{beyondscale2025security} notes that the primary security focus across vendors is access control; encrypting the vectors is not addressed.
Across all major vendors, the standard practice is limited to disk-level and transport-level encryption.
Taken together, query-level encryption of embeddings remains theoretical in academia and absent in industrial deployments.
Developing a practical query-level encryption scheme for vector databases is an important direction for future work.

\end{document}